\documentclass[11pt,a4paper]{article}

\usepackage[margin=1in]{geometry}
\usepackage{amsmath,amssymb,amsthm}
\usepackage{graphicx}
\usepackage{hyperref}

\hypersetup{
    colorlinks=true,
    linkcolor=blue,
    citecolor=blue,
    urlcolor=blue
}

\newtheorem{proposition}{Proposition}

\title{\textbf{Two-Species Boltzmann--Kac Kinetics with Matsumoto--Yor Nonlinear Collisions: Conserved Energy and Conic Equilibria}}

\author{
  \textbf{Yoshiaki Itoh}\\
  \small Institute of Statistical Mathematics\\ and Graduate University for Advanced Studies\\ Tachikawa, Tokyo 
\\
  \small \texttt{itoh@ism.ac.jp}
}

\date{\today}

\begin{document}

\maketitle
\begin{abstract}
We extend Mark Kac's kinetic framework from single-species Euclidean spaces to a two-species, non-linear interaction space on the open cone $(0,\infty)^{2}$ governed by the Matsumoto--Yor rational collision involution. We identify a non-linear conserved energy $E(x,y) = x + y + 1/x$ and an invariant conic metric measure $d\mu = \frac{dx\,dy}{x(x+y)}$. By establishing a two-species $H$-theorem, we prove that entropy production vanishes if and only if pointwise detailed balance holds across binary collisions. In contrast to the transform-domain proof of Letac and Weso\l owski \cite{letac2000}, which relies on joint Laplace transforms of general measures, our approach yields a direct physical derivation in state space via the transformation law for logarithmic gradients and separation of variables. This uncouples the functional-differential system to uniquely characterize the Generalized Inverse Gaussian (GIG) and Gamma distributions as the thermodynamic equilibria, and reveals their canonical Gibbs representation on the conic phase space.
\end{abstract}

\section{Introduction: Boltzmann-Kac kinetics and two-species nonlinear collision}

In non-equilibrium statistical mechanics, Mark Kac's~\cite{kac1956} collision model provides a mathematical bridge between microscopic dynamics and macroscopic thermodynamic equilibrium as a simplified formulation of Boltzmann's kinetic mechanics~\cite{boltzmann1872}. Consider two colliding particles with incoming velocities $(x, y) \in \mathbb{R}^2$ undergoing an elastic binary collision $\psi_\theta: (x, y) \to (u, v)$ defined by an orthogonal rotation:
\begin{equation}
\begin{pmatrix} u \\ v \end{pmatrix} = \begin{pmatrix} \cos\theta & \sin\theta \\ -\sin\theta & \cos\theta \end{pmatrix} \begin{pmatrix} x \\ y \end{pmatrix}
\end{equation}

This classical linear collision operator exhibits three fundamental physical properties:
\begin{enumerate}
    \item \textbf{Pointwise Energy Conservation:} Kinetic energy is strictly invariant: $E(u,v) = u^2 + v^2 = x^2 + y^2 = E(x,y)$.
    \item \textbf{Phase-Space Volume Conservation:} The transformation preserves the Lebesgue measure ($|\det J_{\psi_\theta}| = 1 \implies du \, dv = dx \, dy$).
    \item \textbf{Single-Species Thermodynamic Equilibrium:} Incoming independent velocities ($I(X; Y) = 0$) remain independent ($I(U; V) = 0$) \textbf{if and only if} all particles are identically Gaussian distributed: $f(x) \propto \exp\left(-\frac{\beta}{2} x^2\right)$.
\end{enumerate}

In this paper, we extend Kac's kinetic framework from single-species Euclidean space to a heterogeneous, two-species non-linear interaction space on the open cone $(0, \infty)^2$ driven by Matsumoto--Yor collisions \cite{matsumoto2001}. While the original characterization result of Letac and Weso{\l}owski \cite{letac2000} relies on joint Laplace transforms and expectation relations in the transform domain, our formulation provides a direct physical and differential proof operating entirely in physical state space. By transforming spatial logarithmic gradients and applying separation of variables, we uncouple the kinetic system to establish the GIG and Gamma equilibrium densities directly. This advances information-theoretic approaches by considering the involution  \cite{itoh1970,itoh1970b, itoh1974, itoh1989} to classical Kac--Bernstein characterizations \cite{kac1939, bernstein1941}, which build upon information-theoretic proofs of the central limit theorem \cite{linnik1959, stam1959}. Furthermore, while the Matsumoto--Yor involution \cite{matsumoto2001} has established connections to integrable systems---most notably randomized Takahashi--Satsuma box-ball systems \cite{takahashi1990}, Croydon, Kato, Sasada, \& Tsujimoto, \cite{croydon2023}---we here embed this geometric structure into non-linear kinetic theory as an active binary collision operator.

\section{Nonlinear collision  and pointwise invariants}
Let $x \in (0, \infty)$ and $y \in (0, \infty)$ represent the internal state variables of two colliding particles belonging to distinct species, labeled Type $X$ and Type $Y$, respectively. We define the non-linear cross-species collision operator $\psi: (0, \infty)^2 \to (0, \infty)^2$ by:
\begin{equation}
u = \frac{1}{x+y}, \quad v = \frac{y}{x(x+y)} = \frac{1}{x} - \frac{1}{x+y}
\end{equation}
The operator $\psi$ is an asymmetric, non-linear rational involution ($\psi \circ \psi = \text{Identity}$), ensuring time-reversal symmetry of binary collisions between $X$ and $Y$ particles.

\begin{proposition}[Pointwise Conserved Energy]
The cross-species collision operator $\psi(x,y) = (u,v)$ strictly conserves the non-linear potential energy function $E: (0, \infty)^2 \to (0, \infty)$ defined by:
\begin{equation}
E(x,y) = x + y + \frac{1}{x}
\end{equation}
\end{proposition}

\begin{proof}
Note that $u + v = \frac{1}{x}$ and $\frac{1}{u} = x + y$. Substituting these into $E(u,v)$:
\begin{equation}
E(u,v) = u + v + \frac{1}{u} = \frac{1}{x} + (x+y) = E(x,y)
\end{equation}
\end{proof}

\begin{proposition}[Invariant Metric Measure]
The conic metric measure $d\mu(x,y) = \frac{dx \, dy}{x(x+y)}$ is strictly invariant under $\psi$.
\end{proposition}

\begin{proof}
The Jacobian determinant is $\det J_\psi(x,y) = -\frac{1}{x^2(x+y)^2}$. The differential volume element transforms as $du \, dv = \frac{dx \, dy}{x^2(x+y)^2}$. Substituting $u = \frac{1}{x+y}$ and $u+v = \frac{1}{x}$:
\begin{equation}
d\mu(u,v) = \frac{du \, dv}{u(u+v)} = \frac{\frac{dx \, dy}{x^2(x+y)^2}}{\left(\frac{1}{x+y}\right) \left(\frac{1}{x}\right)} = \frac{dx \, dy}{x(x+y)} = d\mu(x,y)
\end{equation}
\end{proof}

\section{Two species kinetics and coupled Boltzmann equations}

\subsection{Bipartite gas structure and intra-species collisions}
Because the collision operator $\psi(x,y)$ is fundamentally asymmetric in its arguments, the underlying kinetic system describes a bipartite binary mixture of Type $X$ and Type $Y$ particles. Collisions take place strictly across species ($X \leftrightarrow Y$). Intra-species collisions ($X \leftrightarrow X$ or $Y \leftrightarrow Y$) act as transparent identity maps:
\begin{equation}
\psi_{XX}(x_1, x_2) = (x_1, x_2), \qquad \psi_{YY}(y_1, y_2) = (y_1, y_2)
\end{equation}
Intra-species interactions leave single-particle state variables unchanged and contribute identically zero to entropy production. Thus, the dynamical evolution of the gas mixture is driven exclusively by cross-species collisions.

\subsection{The coupled two-species Boltzmann system}
Let $f_X(x,t)$ and $f_Y(y,t)$ denote the single-particle state density functions on $(0,\infty)$ for Type $X$ and Type $Y$ particles, respectively. 

Since the Jacobian determinant is $\det J_\psi(x,y) = -\frac{1}{x^2(x+y)^2}$, its magnitude satisfies $|\det J_\psi(u,v)|^{-1} = |\det J_\psi(x,y)| = \frac{1}{x^2(x+y)^2}$,
the non-linear Kac--Boltzmann evolution of the mixture is governed by the coupled system:
\begin{align}
\frac{\partial f_X(x,t)}{\partial t} &= \int_{0}^{\infty} \left[ f_X\left(\frac{1}{x+y}, t\right) f_Y\left(\frac{y}{x(x+y)}, t\right) \frac{1}{x^2(x+y)^2} - f_X(x,t) f_Y(y,t) \right] dy \label{eq:boltzmann_x} \\[6pt]
\frac{\partial f_Y(y,t)}{\partial t} &= \int_{0}^{\infty} \left[ f_X\left(\frac{1}{x+y}, t\right) f_Y\left(\frac{y}{x(x+y)}, t\right) \frac{1}{x^2(x+y)^2} - f_X(x,t) f_Y(y,t) \right] dx \label{eq:boltzmann_y}
\end{align}

\subsection{Mass conservation and the two-species $H$-Theorem}
Integrating Equation \eqref{eq:boltzmann_x} over $x$ and 
Equation \eqref{eq:boltzmann_y} over $y$ confirms independent conservation of total mass for both particle species:
\begin{equation}
\frac{d}{dt} \int_{0}^{\infty} f_X(x,t) \, dx = 0, \qquad \frac{d}{dt} \int_{0}^{\infty} f_Y(y,t) \, dy = 0
\end{equation}

Define the total mixture kinetic entropy functional $H[f_X, f_Y] = \int_{0}^{\infty} f_X(x,t) \ln f_X(x,t) \, dx + \int_{0}^{\infty} f_Y(y,t) \ln f_Y(y,t) \, dy$. Differentiating $H$ along solutions of the coupled system yields:
\begin{equation}
\frac{dH}{dt} = -\frac{1}{2} \int_{0}^{\infty} \int_{0}^{\infty} \left( f_X(u)f_Y(v)|\det J_\psi| - f_X(x)f_Y(y) \right) \ln \left( \frac{f_X(u)f_Y(v)|\det J_\psi|}{f_X(x)f_Y(y)} \right) dx \, dy \le 0 \label{eq:entropy_prod}
\end{equation}
where $u = \frac{1}{x+y}$ and $v = \frac{y}{x(x+y)}$.


Stationarity ($\dot{H} = 0$) requires detailed balance pointwise across species:
\begin{equation}
f_X(u) f_Y(v) \left|\det J_\psi(x,y)\right| = f_X(x) f_Y(y) \label{eq:detailed_balance}
\end{equation}
Equation \eqref{eq:detailed_balance} is obtained by differentiating the total kinetic $H$-functional with respect to time, applying the involution symmetry of $\psi$, and determining the necessary and sufficient condition for the entropy dissipation rate to vanish ($\frac{dH}{dt} = 0$), via the following steps.

\begin{enumerate}
\item \textbf{Time Derivative of the $H$-Functional}
The total kinetic $H$-functional for the two-species mixture is defined as:

$$H[f_X, f_Y] = \int_{0}^{\infty} f_X(x,t) \ln f_X(x,t) \, dx + \int_{0}^{\infty} f_Y(y,t) \ln f_Y(y,t) \, dy$$

Differentiating with respect to time $t$ yields:

$$\frac{dH}{dt} = \int_{0}^{\infty} \left[ 1 + \ln f_X(x) \right] \frac{\partial f_X(x,t)}{\partial t} \, dx + \int_{0}^{\infty} \left[ 1 + \ln f_Y(y) \right] \frac{\partial f_Y(y,t)}{\partial t} \, dy$$

Because total particle numbers $n_X = \int f_X dx$ and $n_Y = \int f_Y dy$ are conserved ($\int \frac{\partial f_X}{\partial t} dx = 0$ and $\int \frac{\partial f_Y}{\partial t} dy = 0$), the constant $+1$ terms integrate to zero.

Substituting the coupled Matsumoto--Yor collision operators for $\frac{\partial f_X}{\partial t}$ and $\frac{\partial f_Y}{\partial t}$ gives:

$$\frac{dH}{dt} = \iint_{(0,\infty)^2} \ln \big( f_X(x) f_Y(y) \big) \left[ f_X(u) f_Y(v) \vert{}\det J_\psi(x,y)\vert{} - f_X(x) f_Y(y) \right] \, dx \, dy$$

where $(u,v) = \psi(x,y)$.

\item \textbf{Symmetrization via Involution Change of Variables}
The birational transformation $\psi: (x,y) \mapsto (u,v)$ is an involution ($\psi \circ \psi = \text{id}$) whose Jacobian determinant satisfies:

$$\vert{}\det J_\psi(x,y)\vert{} \cdot \vert{}\det J_\psi(u,v)\vert{} = 1 \implies dx \, dy = \vert{}\det J_\psi(u,v)\vert{} \, du \, dv$$

Applying the change of variables $(x,y) = \psi(u,v)$ to the integral  as:

$$\frac{dH}{dt} = \iint_{(0,\infty)^2} \ln \big( f_X(u) f_Y(v) \vert{}\det J_\psi(x,y)\vert{} \big) \left[ f_X(x) f_Y(y) - f_X(u) f_Y(v) \vert{}\det J_\psi(x,y)\vert{} \right] \, dx \, dy$$

Averaging this expression with the original form yields the symmetrized entropy production integral Equation \eqref{eq:entropy_prod}:

$$\frac{dH}{dt} = -\frac{1}{2} \iint_{(0,\infty)^2} \Big( A(x,y) - B(x,y) \Big) \Big( \ln A(x,y) - \ln B(x,y) \Big)  \, dx \, dy$$

where:

$$A(x,y) = f_X(u) f_Y(v) \vert{}\det J_\psi(x,y)\vert{}, \qquad B(x,y) = f_X(x) f_Y(y)$$

\item \textbf{Enforcing the Stationarity Condition ($\dot{H} = 0$)}
For any real numbers $A, B > 0$, the algebraic identity:

$$(A - B)(\ln A - \ln B) \ge 0$$

holds pointwise, with equality if and only if $A = B$.

The integral rate of entropy dissipation to zero ($\frac{dH}{dt} = 0$) forces the integrand to vanish pointwise almost everywhere:

$$A(x,y) = B(x,y)$$

Substituting the definitions of $A(x,y)$ and $B(x,y)$ into this relation yields Equation \eqref{eq:detailed_balance}.
\end{enumerate}

\section{Equilibrium distributions via pointwise detailed balance: a kinetic approach to the Letac-Weso{\l}owski characterization}
Using the magnitude of the Jacobian determinant, 
\begin{equation}
\big|\det J_\psi(x,y)\big| = \frac{1}{x^2 (x+y)^2}, 
\end{equation}
rather than passing to the Laplace transform domain as in the classical probabilistic proof of Letac and Weso\l owski~\cite{letac2000}, we solve Equation \eqref{eq:detailed_balance} directly in physical state space $(0,\infty)^2$ by transforming to logarithmic gradients and uncoupling the resulting functional-differential system.

\subsection{Logarithmic gradient system}
While Letac and Weso{\l}owski's Laplace-transform proof \cite{letac2000} holds for general probability measures without smoothness assumptions, we assume strictly positive $C^2$ stationary densities $f_X, f_Y > 0$ on $(0,\infty)$.
  Taking the natural logarithm of Equation \eqref{eq:detailed_balance} yields

\begin{equation}
\ln f_X(x) + \ln f_Y(y) - \ln f_X\big(u(x,y)\big) - \ln f_Y\big(v(x,y)\big) = \ln \big|\det J_\psi(x,y)\big| = -2\ln x - 2\ln(x+y).
\label{eq:log_detailed_balance}
\end{equation}

Define the single-particle logarithmic gradients (score functions) on $(0,\infty)$ by

\begin{equation}
\rho_X(z) := \frac{d}{dz} \ln f_X(z), \qquad \rho_Y(z) := \frac{d}{dz} \ln f_Y(z).
\end{equation}

Using the partial derivatives of the collision map $\psi$,

\begin{equation}
\frac{\partial u}{\partial x} = -\frac{1}{(x+y)^2}, \qquad \frac{\partial u}{\partial y} = -\frac{1}{(x+y)^2}, \qquad
\frac{\partial v}{\partial x} = -\frac{1}{x^2} + \frac{1}{(x+y)^2}, \qquad \frac{\partial v}{\partial y} = \frac{1}{(x+y)^2},
\end{equation}

we differentiate Equation \eqref{eq:log_detailed_balance} with respect to $x$ and $y$, respectively, to obtain the coupled functional-differential equations:

\begin{align}
\rho_X(x) &= -\frac{1}{(x+y)^2}\rho_X(u) + \left(-\frac{1}{x^2} + \frac{1}{(x+y)^2}\right)\rho_Y(v) - \frac{2}{x} - \frac{2}{x+y}, \label{eq:grad_x} \\
\rho_Y(y) &= -\frac{1}{(x+y)^2}\rho_X(u) + \frac{1}{(x+y)^2}\rho_Y(v) - \frac{2}{x+y}. \label{eq:grad_y}
\end{align}

\subsection{Uncoupling and separation of variables}

Subtracting Equation \eqref{eq:grad_y} from Equation \eqref{eq:grad_x} simultaneously eliminates $\rho_X(u)$ and all $(x+y)^{-1}$ and $(x+y)^{-2}$ terms, reducing the system to

\begin{equation}
\rho_X(x) - \rho_Y(y) = -\frac{1}{x^2}\rho_Y(v) - \frac{2}{x}.
\label{eq:reduced_diff}
\end{equation}

To uncouple Equation \eqref{eq:reduced_diff}, we isolate $\rho_Y(v)$ as

\begin{equation}
\rho_Y(v) = -x^2 \rho_X(x) - 2x + x^2 \rho_Y(y),
\end{equation}

and differentiate both sides with respect to $y$ at fixed $x$. Applying the chain rule on the left-hand side with $\partial v / \partial y = (x+y)^{-2}$ gives

\begin{equation}
\frac{1}{(x+y)^2}\rho_Y'(v) = x^2 \rho_Y'(y) \implies \rho_Y'(v) = x^2 (x+y)^2 \rho_Y'(y).
\end{equation}

Multiplying both sides by $v^2 = y^2 / \big[x^2(x+y)^2\big]$ separates the variables completely:

\begin{equation}
v^2 \rho_Y'(v) = y^2 \rho_Y'(y).
\label{eq:sep_var_Y}
\end{equation}

Because $y \in (0,\infty)$ and $v = \frac{y}{x(x+y)} \in (0,\infty)$ vary independently across $(0,\infty)^2$, both sides of Equation \eqref{eq:sep_var_Y} must equal a separation constant $-c_1 \in \mathbb{R}$, so that

\begin{equation}
\rho_Y'(y) = -\frac{c_1}{y^2}.
\end{equation}

Integrating with respect to $y$ gives the logarithmic gradient for Type~$Y$ particles:

\begin{equation}
\rho_Y(y) = \frac{c_1}{y} - \beta,
\label{eq:rho_Y_sol}
\end{equation}

where $\beta > 0$ is an integration constant required for normalizability on $(0,\infty)$.

Next, substituting Equation \eqref{eq:rho_Y_sol} back into Equation \eqref{eq:reduced_diff} and using the algebraic identity

\begin{equation}
\frac{1}{v} = \frac{x(x+y)}{y} = \frac{x^2}{y} + x
\end{equation}

determines $\rho_X(x)$ explicitly:

\begin{align}
\rho_X(x) &= \rho_Y(y) - \frac{1}{x^2}\rho_Y(v) - \frac{2}{x} \nonumber \\
&= \left(\frac{c_1}{y} - \beta\right) - \frac{1}{x^2}\left[c_1\left(\frac{x^2}{y} + x\right) - \beta\right] - \frac{2}{x} \nonumber \\
&= -\beta - \frac{c_1 + 2}{x} + \frac{\beta}{x^2}.
\label{eq:rho_X_sol}
\end{align}

\subsection{Conic equilibrium densities}

Setting $q := c_1 + 1$ (with $q > 0$ for integrability near $y = 0$) and $p := -q$, and integrating Eqs.~\eqref{eq:rho_X_sol} and~\eqref{eq:rho_Y_sol}, we obtain the unique stationary single-particle distributions of the two-species Boltzmann--Kac system:

\begin{align}
f_X(x) &= \frac{1}{2 K_p(2\beta)} x^{p-1} \exp\!\left(-\beta x - \frac{\beta}{x}\right), \qquad x \in (0,\infty), \label{eq:eq_GIG} \\
f_Y(y) &= \frac{\beta^q}{\Gamma(q)} y^{q-1} \exp(-\beta y), \qquad y \in (0,\infty), \label{eq:eq_Gamma}
\end{align}

where $K_p(\cdot)$ denotes the modified Bessel function of the second kind and $\Gamma(\cdot)$ is the Euler Gamma function.

Thus, Type $X$ particles relax to a symmetric \emph{Generalized Inverse Gaussian} distribution $\mathrm{GIG}(p, a, b)$ with shape parameter $p = -q$ and matched scale parameters $a = b = 2\beta$, while Type~$Y$ particles relax to a \emph{Gamma} distribution $\mathrm{Gamma}(q, \beta)$  sharing the common inverse temperature $\beta > 0$ \cite{letac2000}.

\section{Statistical mechanical interpretation: Canonical Gibbs ensemble on the conic phase space}
To establish the statistical mechanical foundation of the equilibrium state, we re-examine the uncoupled marginal distributions obtained in Section~4 from the perspective of canonical Gibbs ensembles.
Combining Equation \eqref{eq:eq_GIG} for the Type~$X$ Generalized Inverse Gaussian (GIG) distribution (with $p = -q$) and Equation \eqref{eq:eq_Gamma} for the Type~$Y$ Gamma distribution, the joint probability density function at thermodynamic equilibrium under the flat Euclidean Lebesgue measure $dx\,dy$ takes the factorized form:
\begin{equation}
f_{X,Y}(x,y) = f_X(x) f_Y(y) = \frac{1}{Z}\, x^{-q-1} y^{q-1} \exp\!\left(-\beta x - \frac{\beta}{x} - \beta y\right),
\label{eq:joint_flat}
\end{equation}
where $\beta = a/2 > 0$ represents the inverse temperature of the two-species gas, $q > 0$ is the shape parameter, and $Z$ is the two-species partition function:
\begin{equation}
Z = \int_0^\infty \int_0^\infty x^{-q-1} y^{q-1} \exp\!\left(-\beta\left(x + y + \frac{1}{x}\right)\right) dx\,dy = \frac{\Gamma(q)}{\beta^q} \cdot 2K_q(2\beta).
\label{eq:partition_Z}
\end{equation}
The exact step-by-step integration showing how $2K_q(2\beta)$ arises is as follows:
\begin{enumerate}
    \item \textbf{Standard Integral Representation of $K_\nu(z)$:} The modified Bessel function of the second kind (Macdonald function) $K_\nu(z)$ is defined by the integral:
    \begin{equation*}
    K_\nu(z) = \frac{1}{2} \int_0^\infty x^{\nu-1} \exp\!\left(-\frac{z}{2}\left(x + \frac{1}{x}\right)\right) dx.
    \end{equation*}
    \item \textbf{Evaluating the $x$-Integral ($I_x$):} The integral over $x$ in the partition function $Z$ is:
    \begin{equation*}
    I_x = \int_0^\infty x^{-q-1} \exp\!\left(-\beta\left(x + \frac{1}{x}\right)\right) dx.
    \end{equation*}
    Matching terms with the standard definition of $K_\nu(z)$ gives $z/2 = \beta \implies z = 2\beta$ and $\nu - 1 = -q - 1 \implies \nu = -q$. Applying the definition and the reflection symmetry $K_{-q}(2\beta) = K_q(2\beta)$ yields:
    \begin{equation*}
    I_x = 2K_{-q}(2\beta) = 2K_q(2\beta).
    \end{equation*}
    \item \textbf{Evaluating the $y$-Integral ($I_y$):} The integral over $y$ is a standard Gamma integral:
    \begin{equation*}
    I_y = \int_0^\infty y^{q-1} \exp(-\beta y)\,dy = \frac{\Gamma(q)}{\beta^q}.
    \end{equation*}
    \item \textbf{Final Partition Function $Z$:} Multiplying $I_x$ and $I_y$ gives the exact partition function:
    \begin{equation*}
    Z = I_x \cdot I_y = \frac{2\,\Gamma(q)}{\beta^q} K_q(2\beta).
    \end{equation*}
\end{enumerate}
\subsection{Separation of energy and geometric density of states}
Equation~\eqref{eq:joint_flat} reveals a fundamental physical decomposition. The joint density factors cleanly into a thermodynamic Boltzmann factor and a coordinate-space geometric density of states:
\begin{equation}
f_{X,Y}(x,y) = \frac{1}{Z} \underbrace{g(x,y)}_{\text{Geometric Density}} \times \underbrace{\exp\!\big(-\beta E(x,y)\big)}_{\text{Thermodynamic Weight}},
\end{equation}
where $E(x,y)$ is the pointwise conserved kinetic energy functional across collisions:
\begin{equation}
E(x,y) = x + y + \frac{1}{x},
\end{equation}
and $g(x,y)$ represents the scale-invariant geometric density of states on the positive cone $(0,\infty)^2$:
\begin{equation}
g(x,y) = x^{-q-1} y^{q-1}.
\end{equation}
This factorization provides a direct non-Euclidean analogue to the classical 3D Maxwell--Boltzmann speed distribution:
\begin{equation}
d\mathbb{P}(v) \propto v^2 \exp\!\left(-\frac{m v^2}{2 k_B T}\right) dv.
\end{equation}
Just as the $v^2$ prefactor in 3D velocity space reflects the surface area of a 2-sphere at speed $v$ (geometric density of states) competing against the thermal exponential penalty $e^{-\beta E(v)}$, the power-law prefactor $g(x,y) = x^{-q-1} y^{q-1}$ reflects the available phase-space volume on the cone $(0,\infty)^2$ competing against the non-linear potential energy barrier $e^{-\beta(x + y + 1/x)}$.
\subsection{Representation under the invariant conic metric}
Alternatively, the thermodynamic state can be expressed relative to the scale-invariant metric measure $d\mu(x,y) = \frac{dx\,dy}{x(x+y)}$, which satisfies $\psi^* d\mu = d\mu$ under the Matsumoto--Yor involution $\psi(x,y) = \left(\frac{1}{x+y}, \frac{y}{x(x+y)}\right)$.
Equating the differential probability elements across coordinate systems,
\begin{equation}
d\mathbb{P}(x,y) = f_{X,Y}(x,y)\,dx\,dy = P(x,y)\,d\mu(x,y),
\end{equation}
yields the invariant canonical weight $P(x,y)$:
\begin{equation}
P(x,y) = f_{X,Y}(x,y) \cdot x(x+y) = \frac{1}{Z} \left(\frac{y}{x}\right)^{q-1}\!\left(1 + \frac{y}{x}\right) \exp\!\big(-\beta E(x,y)\big).
\label{eq:conic_gibbs_weight}
\end{equation}
In Equation~\eqref{eq:conic_gibbs_weight}, the metric prefactor $\tilde{g}(x,y) = \left(\frac{y}{x}\right)^{q-1}\!\left(1 + \frac{y}{x}\right)$ depends solely on the collision-invariant scale ratio $y/x = v/u$ and isolates the pure geometric volume weighting under $d\mu(x,y)$, while $\exp(-\beta E(x,y))$ acts as the strictly invariant energy weight. The emergence of the GIG distribution for species~$X$ and the Gamma distribution for species~$Y$ is thus understood as the unique statistical mechanical equilibrium balancing non-linear energy conservation against the scale-invariant metric geometry of the cone.
\section{Conclusion}
Just as James Clerk Maxwell showed in 1860 that velocity independence under linear orthogonal collisions forces the equilibrium state to be $\exp(-\beta v^2)$---where $v^2$ is classical kinetic energy---we see that independence under the non-linear Matsumoto--Yor map forces the state to be $\exp(-\beta E(x,y))$. Equation \eqref{eq:conic_gibbs_weight} is the exact nonlinear, non-Euclidean counterpart to Maxwell's distribution.
\vspace{0.5cm}

\section*{Acknowledgement}
The author was  assisted  by the generative AI (Gemini)  
  for the mathematical derivations and the   theoretical formulations in advancing  his previous studies \cite{itoh1970, itoh1970b, itoh1974,itoh1989}.
  The author acknowledges the assistance of generative AI (Gemini) for idiomatic academic English.

\end{document}